\documentclass{elsarticle}

\usepackage{hyperref}
\usepackage{xcolor,amssymb}
\usepackage[utf8]{inputenc}

\usepackage{booktabs}
\usepackage{multirow}
\usepackage{color}
\usepackage{changes}

\journal{arXiv}
\usepackage[english]{babel}	
\usepackage{microtype}

\usepackage[T1]{fontenc}
\usepackage[utf8]{inputenc}

\usepackage{amsmath,amssymb,amsthm}
\usepackage{mathtools}

\usepackage[linesnumbered,ruled,vlined]{algorithm2e}

\newcounter{cntr}
\theoremstyle{definition}

\newtheorem{remark}[cntr]{Remark}
\newtheorem{problem}[cntr]{Problem}

\theoremstyle{plain}
\newtheorem{theorem}[cntr]{Theorem}
\newtheorem{lemma}[cntr]{Lemma}

\newtheorem{algo}[cntr]{Algorithm}

\newcommand{\bb}[1]{\ensuremath{\mathbf{#1}}}
\newcommand{\C}{\mathbb{C}}
\newcommand{\R}{\mathbb{R}}
\newcommand{\curl}{{\pmb{\operatorname{curl}}}{}}
\renewcommand{\div}{\operatorname{div}{}}
\newcommand{\grad}{\operatorname{grad}{}}
\newcommand{\MSL}{{\pmb {\mathcal V}_\kappa}{}}
\newcommand{\MSLP}{{\pmb {V}_\kappa}{}}
\newcommand{\MSLh}{{\pmb {\mathrm V}_{\kappa}}{}}

\usepackage{pgfplots}
\usepackage{pgfplotstable}
\usepackage{amsmath}

\newcommand{\convfitorder}[4]{
		\pgfplotstableread{#2}\meta
		\pgfplotstablegetrowsof{\meta}
		\pgfmathsetmacro{\N}{\pgfplotsretval-1}
		\pgfplotstablegetelem{\N}{M}\of{\meta}
		\let\lastM\pgfplotsretval
		\pgfplotstablegetelem{\N}{#1}\of{\meta}
		\let\lasterror\pgfplotsretval
		
		\addplot[thick,dashed,mark = none] table [trim cells=true,x=M,y expr=2^(-#3*(\thisrow{M}-\lastM))*\lasterror*1.3] {#2};
		\addlegendentry{#4}
	}

\pgfplotsset{width=.8\textwidth}

\pgfplotsset{compat=newest}
\pgfplotsset{plot coordinates/math parser=false} 

\usetikzlibrary{external, calc}
\usepgfplotslibrary{colorbrewer}

\pgfplotsset{colormap/Dark2}
\pgfplotsset{
	cycle list/Dark2,
	cycle multiindex* list={
		mark list*\nextlist
		Dark2\nextlist
	},
}

\pgfplotsset{%
		lnc/.style={%
			line width=1.2,
			mark={},
			},
	}

\usepackage{expl3}
\ExplSyntaxOn
\cs_set_eq:NN \fpeval \fp_eval:n
\ExplSyntaxOff

\begin{document}
\begin{frontmatter}
\title{Solving Maxwell’s Eigenvalue Problem via Isogeometric Boundary Elements and a Contour Integral Method}
\author[aff1]{Stefan~Kurz}
\ead{kurz@gsc.tu-darmstadt.de}
\author[aff1]{Sebastian~Schöps}
\ead{schoeps@temf.tu-darmstadt.de}
\author[aff2]{Gerhard~Unger}
\ead{gunger@math.tugraz.at}
\author[aff1]{Felix~Wolf{\,}\corref{corr}}
\ead{wolf@temf.tu-darmstadt.de}
\address[aff1]{Technische Universität Darmstadt, Centre of Computational Engineering}
\address[aff2]{Graz University of Technology, Institute of Applied Mathematics}
\cortext[corr]{Corresponding author}
\begin{abstract}
We solve Maxwell's eigenvalue problem via isogeometric boundary elements and a contour integral method. We discuss the analytic properties of the discretisation, outline the implementation, and showcase numerical examples.  
\end{abstract}
\end{frontmatter}

\section{Motivation}

The development and construction of particle accelerators is arguably one of the most time and money consuming research projects in modern experimental physics. 
A cause for this is that essential components are not available off the shelf and must be manufactured uniquely tailored to the design specification of the planned accelerator.

One of the most performance critical components are so-called \emph{cavities}, resonators often made out of superconducting materials in which electromagnetic fields oscillate at {radio frequencies. The resonant fields are then used to accelerate bunches of particles up to speeds close to the speed of light.}
The geometry, and consequently the resonance behaviour of these structures is vital to the overall performance of the accelerator as a whole. 

Due to the expensive (e.g. superconducting) materials and vast amounts of manual labor that are needed in the manufacturing of these devices, the design of cavities has become its own area of research, cf.~\cite{add_Tesla} and the sources cited therein.
Consequently, the development of simulation tools specifically for this purpose became an important part of the related scientific advancement, see e.g.~\cite{Halbach_1976aa,Rienen_1985aa,Weiland_1985aa}.

 \begin{figure}\centering
        \begin{tikzpicture}[scale=1]
            \node (A) at (.5,0) {\includegraphics[width=5cm]{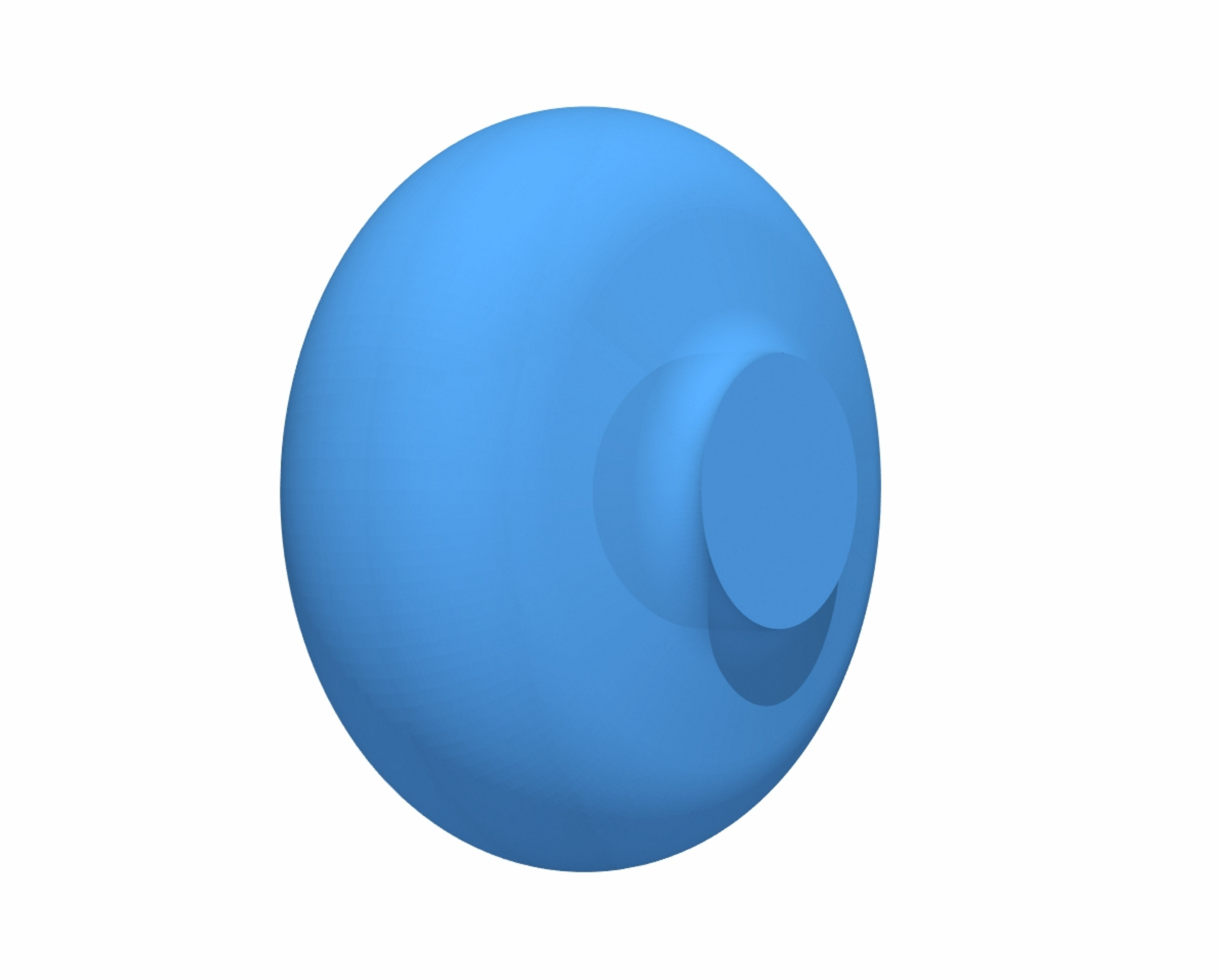}};
            \node (C) at (4.25,0) {\includegraphics[width=5cm]{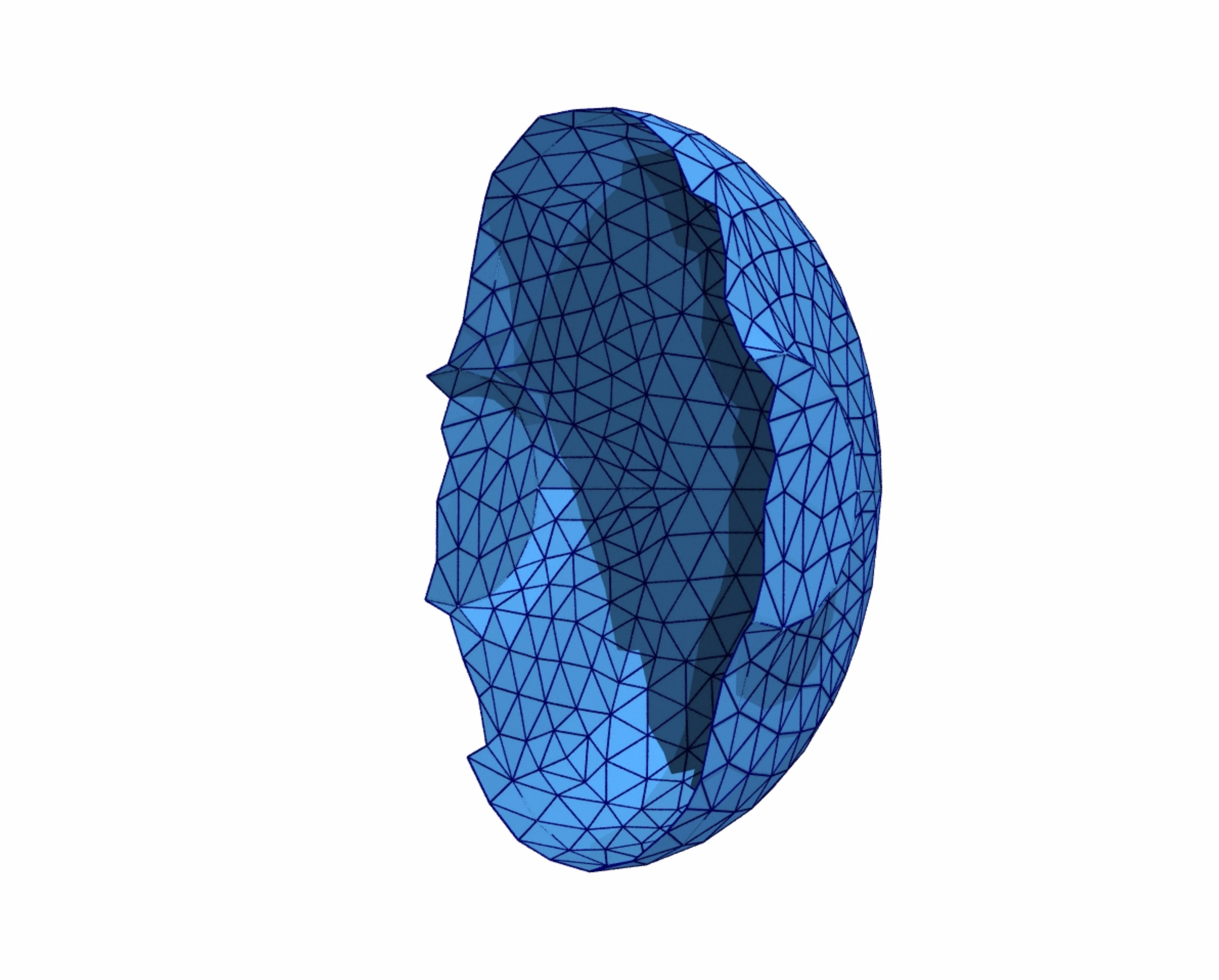}};
            \node (D) at (8,0) {\hspace{.2cm}\includegraphics[width=5cm]{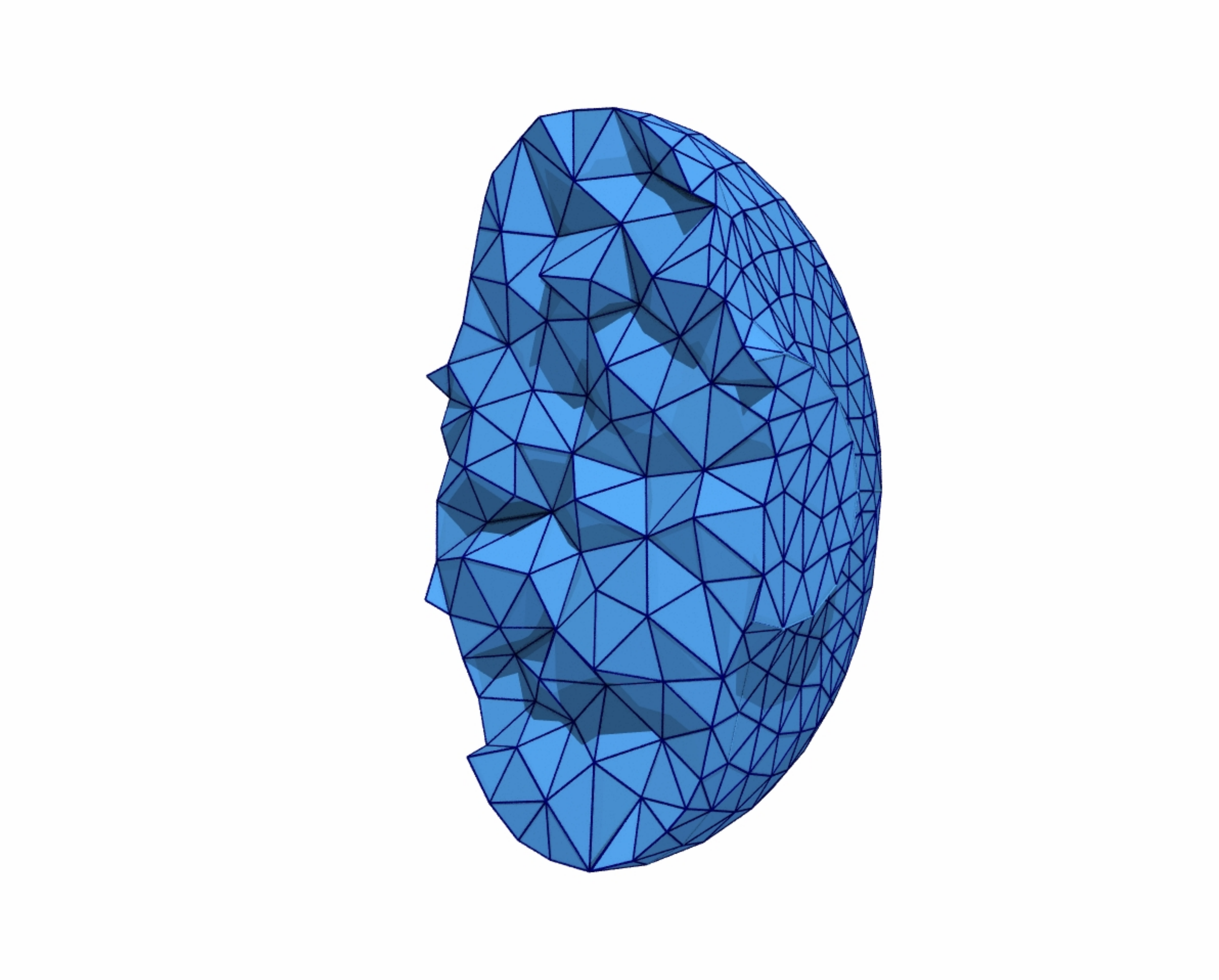}};
            \draw [-latex] (2.05,0) -- (3.05,0);
            \draw [-latex] (5.75,0) -- (6.75,0);
        \end{tikzpicture}
        \caption{For a finite element computation a mesh is generated from the boundary data available from the design framework. Afterwards, a volume mesh is created. This introduces geometrical errors and limits the obtainable order of convergence to that of the geometry representation. Graphics from \cite{Wolf2020}.}\label{fig::intro::tesla1}
    \end{figure}

A bottleneck with these classical approaches has always been the representation of the geometry, which often limits the achievable accuracies, cf. Figure~\ref{fig::intro::tesla1}. {However, high accuracies are desired such that the initial design and its simulation are not the weakest link within the manufacturing pipeline.}
As an example, manufacturers alone are interested in the simulation of deformation effects, the so-called \emph{Lorentz detuning}, which are dependent on a relative error margin of roughly $10^{-7}$, \cite[Tab. II]{add_Tesla}. There are also other, more advanced applications which have such high demands on accuracy. One is presented by Georg et al.~\cite{Georg2019}, {who show that even higher accuracies than those already achievable are required to simulate eccentricities.} 

{\section{Introduction}}

Nowadays, isogeometric analysis \cite{Hughes_2005aa} has been established as the method of choice when dealing with such high demands on accuracy w.r.t.~geometry representation. Isogeometric methods are well understood for the case of electromagnetism \cite{Buffa_2013aa,Bontinck_2017ag} and a corresponding finite element approach has already been applied to Maxwell's eigenvalue problem \cite{Corno_2016aa}.

 \begin{figure}\centering
        \begin{tikzpicture}[scale=1]
            \node (A) at (0,0) {\includegraphics[width=5cm]{figures/1.pdf}};
            \node (B) at (4,0) {\hspace{-.2cm}\includegraphics[width=5cm]{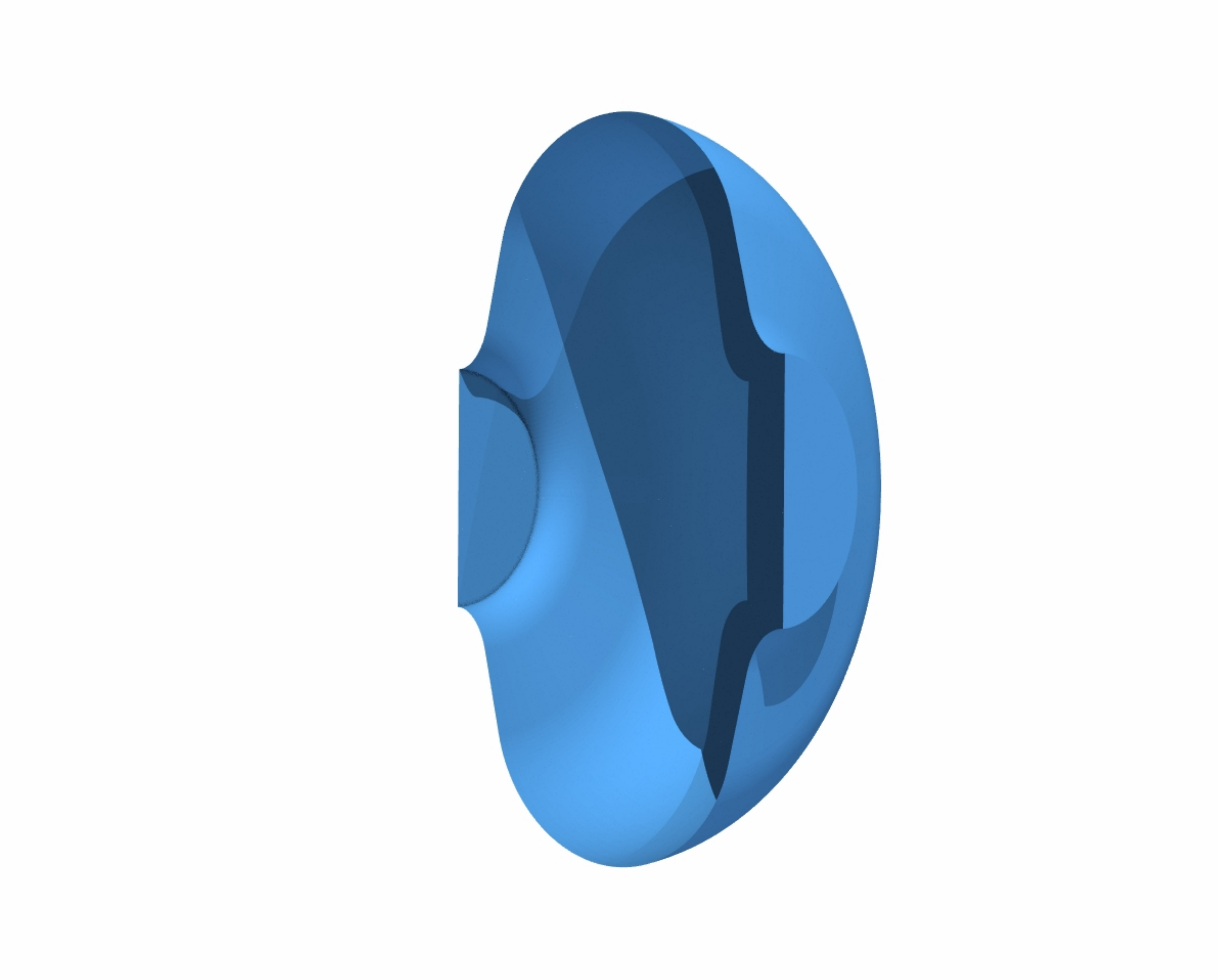}};
            \draw [-latex] (1.7,0) -- (2.7,0);
        \end{tikzpicture}
        \caption{Isogeometric boundary element methods enable the computation directly on the CAD representation. Graphics from \cite{Wolf2020}.}\label{fig::intro::tesla2}    
    \end{figure}

The boundary-only representations in modern CAD frameworks, as well as the demand for highly accurate simulation techniques, suggests the transition to isogeometric boundary element methods for these types of problems, cf.~Figure~\ref{fig::intro::tesla2}.
While the use of boundary element methods promises to both reduce the number of degrees of freedom w.r.t.~the element size $h$ drastically, and at the same time double the rate of convergence {for the point evaluation of the solution in the domain}~\cite[Cor.~3.11 \& Rem.~3.12]{SISC}, 
the system matrices become densely pouplated and the corresponding eigenvalue problem becomes non-linear. However, through the recent introduction of a new family of eigenvalue solvers known as \emph{contour integral methods} \cite{Beyn_2012aa,Asakura_2009aa}, this problem was mitigated.

A first approach to a boundary element eigenvalue problem via the contour integral method 
was investigated {in \cite{WienersXin:2013}, however, neither with a higher-order approach, a discussion of the related convergence theory, or within the isogeometric setting. Moreover, a comparison of volume-based and boundary element methods has not been discussed within the literature.}

This article is build on top of the recent mathematical results of \cite{SISC,Beyn_2012aa,NuMa,UngerPreprint} and aims at advancing the solution of eigenvalue problems via boundary elements by discussing the convergence analysis of the isogeometric discretisation of the eigenvalue problem. In this, we prove {that a convergence order  ${\mathcal O}(h^{2p+1})$ for the eigenvalue aproximation can be achieved for discretisations with mesh-size $h$ and ansatz functions of order  $p$. For a corresponding volume-based IGA  only a  convergence order  ${\mathcal O}(h^{2p})$ can be expected~\cite{BuffaRivasSangalliVazquez:2011}.}  

The organisation of this document is straight forward. Section \ref{sec::evp} will introduce the cavity problem based on Maxwell's equations. Afterwards, Section \ref{sec::efie} will show how this problem can be rephrased as a boundary integral equation, the well-known \emph{electric field integral equation}, and we will discuss the equivalence of both formulations.
We will then discuss our discretisation scheme, where we will review our isogeometric discretisation of the boundary integral equation in Section \ref{sec::IGA}, making theoretical predictions regarding the convergence behaviour. Our approach to the arising non-linear eigenvalue problem will then be discussed in Section \ref{sec::cim}.
In Section \ref{sec::num} a selection of numerical examples will be presented, and finally we will briefly conclude our findings in Section \ref{sec::con}.\\


\section{The Eigenvalue Problem}\label{sec::evp}
For a compact and simply connected Lipschitz domain $\Omega \in \R^3$ the electromagnetic fields in the source-free case are governed by the equations
\begin{equation}\label{eq:Maxwell-time-harmonic}
\begin{aligned}
\curl{(\bb{E})}       &= -i\omega\mu_0\bb{H} && \text{in }\Omega\\
\curl{(\bb{H})}       &= i\omega\varepsilon_0\bb{E} && \text{in }\Omega\\
\div{(\varepsilon_0\bb{E})} &= 0                     && \text{in }\Omega\\
\div{(\mu_0\bb{H})}  &= 0                     && \text{in }\Omega,
\end{aligned}
\end{equation}
assuming the time-harmonic case.
As usual, $\bb E$ and $\bb H$ denote the electric and magnetic field \cite{Jackson_1998aa}, respectively.
In the case of the cavity problem, the electric permittivity $\varepsilon_0$
and the magnetic permeability $\mu_0$ are those of vacuum.
Moreover, since superconducting alloys can be modelled as perfect electric conducting, we assume the boundary conditions on $\Gamma\coloneqq \partial \Omega$ to be given by
\begin{equation}\label{eq:pecBC}
\begin{aligned}  
\bb{E}\times\bb{n} &= 0 && \text{on }\partial\Omega\\
\bb{H}\cdot\bb{n}  &= 0 && \text{on }\partial\Omega.
\end{aligned}
\end{equation}
By eliminating $\bb{H}$ from~\eqref{eq:Maxwell-time-harmonic} one can then derive the classical cavity problem \cite{Monk_2003aa}.
\begin{problem}[Cavity problem]
Find the wave number $k:=\omega \sqrt{\mu_0\epsilon_0} \in \R$ and $\bb{E} \neq 0$ such that
\begin{equation}\label{eq:Maxwell-eig-cont}
\begin{aligned}
\curl{\bigl(\curl{(\bb{E})}\bigr)} &= k^2 \bb{E} && \text{in }\Omega\\
\div{(\bb{E})} &= 0 && \text{in }\Omega\\
\bb{E}\times\bb{n} &= 0 && \text{on }\Gamma.
\end{aligned}
\end{equation}
\end{problem}

\section{The Electric Field Integral Equation}\label{sec::efie}

Before we introduce the variational formulation of Maxwell's eigenvalue problem, we need to introduce some notation.
\subsection{Function Spaces Related to Maxwell's Equations and the EFIE}
Let $\Omega$ be a Lipschitz domain.
By $H^0(\Omega)$ we denote the usual square integral functions $L^2(\Omega)$, and we write $H^s(\Omega)$ for the usual Sobolev spaces of higher regularity $s>0$, cf.~\cite{McLean_2000aa}. Their vector-valued counterparts $\pmb H^s (\Omega) = \big(H^s(\Omega)\big){}^3$ are denoted by bold letters. {We now define the spaces
\begin{align*}
	\pmb H^s(\curl,\Omega){}&{}=\lbrace \pmb f \in \pmb H^s\colon \curl\pmb f \in \pmb H^s(\Omega)\rbrace, \\
	\pmb H^s(\curl^2,\Omega){}&{}=\lbrace \pmb f \in \pmb H^s(\curl,\Omega)\colon \curl\,\curl\pmb f \in \pmb H^s(\Omega)\rbrace,\text{ and}\\
	\pmb H^s(\div,\Omega){}&{}=\lbrace \pmb f \in \pmb H^s\colon \div\pmb f \in H^s(\Omega)\rbrace,
\end{align*}
dropping the index $s$ in the case of $s=0$ for convenience.}
For $\bb{n}_{\pmb x_0}$ denoting the outwards directed unit normal at $\pmb x_0\in \Gamma$, the \emph{rotated tangential trace}
$$\pmb\gamma_{\mathrm t} (\pmb f) \coloneqq \lim_{\pmb x\to \pmb x_0} \pmb f(\pmb x)\times \bb{n}_{\pmb x_0}$$
is well defined for smooth $\pmb f.$ Note that $\pmb n_{\pmb x_0}$ is well-defined for almost all $\pmb x_0,$ cf.~\cite[p.~96]{McLean_2000aa}. We will equip the rotated tangential trace with the superscript~${}^\mathrm{int}$ if the limit is taken from within $\Omega$, and with the superscipt ${}^\mathrm{ext}$ if taken from $\mathbb R^3\setminus\overline{\Omega},$ omitting the notation if the mapping properties are clearly stated.
The operator $\pmb \gamma_{\mathrm t}$ is extended to a weak setting through density arguments.
We can now define
$\pmb H_\times^{-1/2}(\div_\Gamma,\Gamma) = \pmb \gamma_{\mathrm t}\big(\pmb H(\curl,\Omega)\big).$
By definition this renders the rotated tangential trace $ \pmb \gamma_{\mathrm t}\colon \pmb H(\curl,\Omega)\to \pmb H_\times^{-1/2}(\div_\Gamma,\Gamma)$ surjective, and one can prove that it is continuous, compare \cite{Buffa_2003aa}. 

We can now reformulate Maxwell's eigenvalue problem as a variational problem exclusively on the boundary $\Gamma$ by using the anti-symmetric pairing

\begin{align}
\langle \pmb \nu,\pmb \mu\rangle_\times \coloneqq\int_\Gamma (\pmb \nu \times \bb{n})\pmb \mu\mathrm{d} \Gamma,\label{found::eq::dualitypairingX}
\end{align}
cf.~\cite[Def.~1]{Buffa_2003aa}.

\subsection{Recasting the Eigenvalue Problem}
Any solution of the electric wave equation can be derived via a boundary integral formulation, which we will review within this section. 
The version reviewed here resembles the one presented in \cite[Thm.~6]{Buffa_2003aa}.

We define the \emph{Maxwell single layer potential} $$\MSLP\colon \bb H^{-1/2}_\times(\div_\Gamma,\Gamma)\to \bb H(\bb\curl^2,\Omega)$$
via the Helmholtz fundamental solution 
\begin{align}
u^*_\kappa(\pmb x,\pmb y)=\frac{e^{-i\kappa\Vert\pmb x-\pmb y\Vert}}{4\pi\Vert \pmb x-\pmb y\Vert} \label{eq::fundamental}
\end{align}
as
\begin{align*}
\MSLP (\pmb \mu)(\pmb x)= \int_\Gamma u^*_\kappa(\pmb x,\pmb y)\pmb \mu (\pmb y)\mathrm{d} \Gamma_{\pmb y}  
 + \frac 1{\kappa^2} \pmb \grad_{\pmb x}\int_\Gamma u^*_\kappa(\pmb x,\pmb y)(\div_\Gamma\circ \pmb u)(\pmb y)\mathrm{d} \Gamma_{\pmb y}.
\end{align*}
The mapping properties are known, see \cite[Thm.~5]{Buffa_2003aa}, where it is also shown that the image of $\MSLP$ is divergence free.
With the help of this operator, one can show the following.

\begin{lemma}[Single Layer Representation, {\cite[Thm.~6]{Buffa_2003aa}}]\label{found::thm::indirectrepresentation}
If $\bb E\in \bb H(\bb\curl^2,\Omega)$ is solution to the electric wave equation \eqref{eq:Maxwell-eig-cont} on $\Omega$, then it can be represented via
\begin{align*}
 \bb E = \MSLP\big({\bb\gamma}_{\bb t}^\mathrm{int}\bb (\curl\;\bb E)\big).
\end{align*}
\end{lemma}

We now define the 
\emph{Maxwell single layer operator} $\MSL = \pmb\gamma_{\mathrm t}^\mathrm{int} \circ \MSLP$. 
Applying the rotated tangential trace to the identity in Lemma \ref{found::thm::indirectrepresentation} allows us to recast the eigenvalue problem \eqref{eq:Maxwell-eig-cont} as a variational problem w.r.t.~the duality pairing \eqref{found::eq::dualitypairingX}. The underlying boundary integral equation 
reads as follows:
 \begin{problem}[Variational Eigenvalue Problem]\label{problem::variational}
     Find a non-zero $\pmb j\in \pmb H_\times^{{-1/2}}(\div_\Gamma,\Gamma)$ and $\kappa>0$ such that
    \begin{align}\label{Eq:EigVarFormulation}       
        \langle \MSL\pmb j,\pmb \mu\rangle_\times = 0 
    \end{align}
    holds for all $\pmb\mu\in \pmb H_\times^{{-1/2}}(\div_\Gamma,\Gamma)$.
\end{problem}
The problem in~\eqref{Eq:EigVarFormulation} is a nonlinear eigenvalue problem with respect to the eigenvalue parameter $\kappa$, since $\kappa$ occurs nonlinearly in the fundamental solution \eqref{eq::fundamental} which builds the kernel of the single layer boundary integral operator $\MSL$. 
The eigenvalue problem formulations~\eqref{eq:Maxwell-eig-cont} and~\eqref{Eq:EigVarFormulation} are equivalent in the following sense:

\begin{lemma}[Equivalence of eigenvalue problems] Let $\kappa\in\R$ and $\kappa>0$. 
 \begin{enumerate}[a)]
  \item If $(\kappa,\bb E)$ is an eigenpair of~\eqref{eq:Maxwell-eig-cont}, then $(\kappa,({\bb\gamma}_{\bb t}^\mathrm{int}\bb (\curl\;\bb E))$ is an eigenpair of~\eqref{Eq:EigVarFormulation}.
  \item If $(\kappa, \pmb j)$ is an eigenpair of~\eqref{Eq:EigVarFormulation}, then $(\kappa, \MSLP\pmb j)$ is an eigenpair of~\eqref{eq:Maxwell-eig-cont}.
 \end{enumerate}
\end{lemma}
\begin{proof}
Assertion a) has been already shown by the derivation of the boundary integral formulation~\eqref{Eq:EigVarFormulation}. For assertion b) note that $\MSLP \pmb j$ is a solution of Maxwell's equation in $\Omega$~\cite[Sect.~4]{Buffa_2003aa} and that $0=\MSL\pmb j=\pmb\gamma_{\mathrm t}^\mathrm{int}\MSLP \pmb j$. It remains to show that $\MSLP \pmb j\neq 0$ in $\Omega$ which follows form the unique solvability of the exterior problem and the jump relations of the single layer potential on the boundary  $\Gamma$, see~\cite[Prop.~2.1(ii)]{UngerPreprint}.
\end{proof}

We want to mention that  the  eigenvalue problem~\eqref{Eq:EigVarFormulation} has in addition to the real eigenvalues also non-real eigenvalues which correspond to the eigenvalues of the exterior eigenvalue problem. We refer to~\cite{UngerPreprint} for an analysis of this kind of eigenvalues.

\section{A Brief Review of Isogeometric Analysis}\label{sec::IGA}

    This section is devoted to a brief  review of isogeometric analysis as required for its utilisation in the context of boundary element methods for electromagnetism. 
	We follow the lines of \cite{Doelz_2017aa,SISC}, which in their turn are based on the works of Buffa et al.~\cite{Buffa_2010aa} aimed at finite element discretisations.

	Let $p\geq 0$ and $m>0$ be integers. While all of the results reviewed in this article are applicable to so-called \emph{$p$-open locally quasi-uniform knot vectors}, cf.~\cite[Assum.~2.1]{Beirao-da-Veiga_2014aa}, we assume all \emph{knot vectors} to be of the form
	$$\Xi^p_m = \big[\underbrace{0,\cdots,0}_{p+1\text{ times}},1/2^m,\cdots ,(2^m-1)/2^m ,\underbrace{1,\cdots,1}_{p+1\text{ times}}\big].$$
	This is introduced only for notational convenience and in agreement with our numerical examples. 
	We will then refer to $m$ as the \emph{refinement level} and define the \emph{mesh size} $h\coloneqq 2^{-m}.$
	We proceed to define B-spline bases via the well-known recursion formula 
	\begin{align*}
		b_i^p(x) & = \frac{x-\xi_i}{\xi_{i+p}-\xi_i}b_i^{p-1}(x) +\frac{\xi_{i+p+1}-x}{\xi_{i+p+1}-\xi_{i+1}}b_{i+1}^{p-1}(x),
	\end{align*}
	anchored for $p=0$ by the locally constant functions $b_i^0(x) = \chi_{[\xi_i,\xi_i+1)},$ cf.~\cite[Sec.~2.2]{Piegl_1997aa}. We define the spline space $S^p_m$ as the span of the B-spline basis defined on $\Xi^p_m.$

	In reference domain, i.e., on $\square\coloneqq(0,1)^2$, we can now define the divergence conforming spline space $\pmb{\mathbb S}^p_m(\square)$ as done in \cite[Sec.~5.5]{Beirao-da-Veiga_2014aa} via
	$$\pmb{\mathbb S}^p_m(\square)\coloneqq (S^p_m\otimes S^{p-1}_m) \times (S^{p-1}_m\otimes S^{p}_m),$$
	for $\otimes$ denoting the tensor product and $\times$ denoting the Cartesian product.

	\subsection{Geometry and Discretisation in the Physical Domain}

    Let a \emph{patch} $\Gamma$ be given by the image of $\square$ under an invertible diffeomorphism $\pmb F\colon \square\to\Gamma\subseteq\mathbb R^3$.
    For $\Omega$ being a Lipschitz domain, define a \emph{multipatch geometry} to be a compact, orientable two-dimensional manifold $\Gamma=\partial \Omega$ invoked via  $\bigcup_{0\leq j <N} \Gamma_j$ by a family of patches $\lbrace \Gamma_j\rbrace_{0\leq j<N},$ $N\in \mathbb N$. 
    The family of diffeomorphisms $\lbrace \pmb F_j \colon \square\hookrightarrow \Gamma_j\rbrace_{0\leq j<N}$ will be called \emph{parametrisation}.

    We assume the $\Gamma_j$ to be disjoint and that for any \emph{patch interface} $D$ of the form $D=\partial\Gamma_{j_0}\cap \partial\Gamma_{j_1}\neq \emptyset$ we require the continuous extensions of the parametrisations $\pmb F_{j_0}$ and $\pmb F_{j_1}$ onto $\overline{\square}=[0,1]^2$ to coincide.

	Within the framework of isogeometric analysis parametrisations will usually be given by tensor products of non-uniform rational B-splines (NURBS), i.e., functions of the form 
	\begin{align*}
        \pmb F_j(x,y)\coloneqq \sum_{0\leq j_1<k_1}\sum_{0\leq j_2<k_2}\frac{\pmb c_{j_1,j_2} b_{j_1}^{p_1}(x) b_{j_2}^{p_2}(y) w_{j_1,j_2}}{ \sum_{i_1=0}^{k_1-1}\sum_{i_2=0}^{k_2-1} b_{i_1}^{p_1}(x) b_{i_2}^{p_2}(y) w_{i_1,i_2}}.
    \end{align*}

    To extenf the definition of $\pmb{\mathbb S}^p_m$ to the physical domain, we resort to an application of the so-called \emph{Piola transformation} \cite[Chap.~6]{Peterson_2006aa}.

    For the case of geometry mappings $\pmb F_j\colon \square\to\Gamma_j$ and $\pmb f\in \pmb H^{-1/2}_\times(\div_\Gamma,\Gamma)$ that allow for point evaluations, its explicit form is given by 
    \begin{align*}
        \iota_{\pmb F_j}(\pmb f)(\pmb x)&\coloneqq \eta(\pmb x)(\pmb {\mathrm d}\pmb F_j)^{-1} (\pmb f\circ \pmb F_j)\big)(\pmb x),&&\pmb x\in \square.
    \end{align*}
    Therein, $\eta$ denotes the \emph{surface measure}
    $$ \eta(\pmb x) \coloneqq  \Vert{\partial_{x_1}\pmb F_j( {\pmb x})\times \partial_{x_2}\pmb F_j( {\pmb x})}\Vert ,$$ 
    and $\pmb {\mathrm d}\pmb F_j$ denotes the Jacobian of $\pmb F_j.$ Although it is not readily invertible, an inverse exists in the sense of a mapping from the two-dimensional tangent space of $\Gamma_j$ to vector fields on $\square,$ as has already been discussed in \cite[Sec.~3.2]{SISC}.
    We remark that it needs not be computed, since all computations can be reduced to computations in the reference domain. Explicit formulae are easily derived, see e.g.~\cite[Sec.~6.3]{Peterson_2006aa}.

    We can now introduce the globally divergence conforming space
    $$ \pmb {\mathbb S}^p_m(\Gamma)\coloneqq  \left\lbrace \pmb f\in \pmb H_\times^{{-1/2}}(\div_\Gamma,\Gamma)\colon  \iota_{\pmb F_j}(\pmb f|_{\Gamma_j}) \in \pmb {\mathbb S}^p_m(\square)\text{ for all }0\leq j< N\right\rbrace.$$
    It has been analysed in \cite[Def.~10]{NuMa}, where it has been shown that it enjoys quasi-optimal approximation properties w.r.t.~energy norm of the EFIE. Specifically, the spline space satisfies estimates of the kind
    \begin{align}
    	\min_{\pmb f_h\in \pmb{\mathbb S}^p_m(\Gamma)}\Vert{\pmb f-\pmb f_h}\Vert_{\pmb H_\times^{{-1/2}}(\div_\Gamma,\Gamma)}\leq Ch^{s+1/2}\Vert \pmb f\Vert_{\pmb H^s_{\mathrm{pw}}(\div_\Gamma,\Gamma)},\qquad 0\leq s\leq p,\label{eq::iga::approximation}
    \end{align}
    for all $\pmb f\in \pmb H^s_{\mathrm{pw}}(\div_\Gamma,\Gamma)$, cf.~\cite[Thm.~3]{NuMa}, where the space ${\pmb H^s_{\mathrm{pw}}(\div_\Gamma,\Gamma)}$
    is equipped with the norm
    $$\Vert \pmb f \Vert_{\pmb H^s_{\mathrm{pw}}(\div_\Gamma,\Gamma)}\coloneqq \sum_{0\leq j<N}\big\Vert{{\iota_{\pmb F_j}(\pmb f|_{\Gamma_j})\big\Vert_{\pmb H^s(\div,\square)}}}.$$
    We remind of the relation $h=2^{-m}$ due to the assumption of uniform refinement on the knot vectors.
    These discrete spaces give rise to the discrete analogue of Problem \ref{problem::variational}. 
    \begin{problem}[Discrete Eigenvalue Problem]\label{problem::variational::disc}
     Find a non-zero $\pmb j_h\in \pmb{\mathbb S}^p_m(\Gamma)$ and $\kappa_h$ such that
    \begin{align}\label{GalerkinEVP}       
        \langle \MSL\pmb j_h,\pmb \mu\rangle_\times = 0 
    \end{align}
    holds for all $\pmb\mu\in \pmb{\mathbb S}^p_m(\Gamma)$.
    \end{problem}
    
    In its discrete form, this problem induces a linear system
    \begin{align}
    \MSLh \bb j_h = 0,\label{eq::linear_system}
    \end{align}
    which we choose to assemble as in \cite{SISC,IEEE}.
    Therein, $\MSLh$ can be interpreted as a matrix valued function dependent on $\kappa.$

    The well-posedness of this discrete problem is closely related to the following criteria, cf.~\cite[Sec.~3]{Bespalov_2010aa}.

    \begin{itemize}
        \item[(C1)] There exists a continuous splitting $\pmb H^{-1/2}_\times(\div_\Gamma,\Gamma) = \pmb W\oplus \pmb V$ such that the bilinear form $\langle \MSL \cdot,\cdot\rangle_\times$ is stable and coercive on $\pmb V\times \pmb V$ and $\pmb W\times \pmb W$, and compact on $\pmb V\times \pmb W$ and $\pmb W\times \pmb V$.
        \item[(C2)] $\pmb {\mathbb S}^p_m(\Gamma)$ can be decomposed into a sum $\pmb {\mathbb S}^p_{\pmb m}(\Gamma) = \pmb W_h \oplus \pmb V_h$ of closed subspaces of $\pmb H^{-1/2}_\times(\div_\Gamma,\Gamma)$.
        \item[(C3)] $\pmb W_h$ and $\pmb V_h$ are stable under complex conjugation.
        \item[(C4)] Both $\pmb W_h\subseteq \pmb W$, as well as the so-called \emph{gap-property}
            \begin{align}\label{eq::gap-property}
            \sup_{\pmb v_h\in \pmb V_{h}}\inf_{\pmb v\in \pmb V}\frac{\Vert{\pmb v-\pmb v_h}\Vert_{\pmb H^{-1/2}_\times(\div_\Gamma,\Gamma)}}{\Vert{\bb v_h}\Vert_{\pmb H^{-1/2}_\times(\div_\Gamma,\Gamma)}}\stackrel{h\rightarrow 0}{\longrightarrow} 0
            \end{align}
        hold.
    \end{itemize}

    These properties have been proven for the isogeometric discretisation \cite[Thm.~3.9]{SISC}.

    \subsection{Numerical Analysis}

The convergence of the eigenvalues and eigenfunctions of the Galerkin eigenvalue problem~\eqref{GalerkinEVP} can be shown using abstract results of~\cite{Halla:2016,UngerPreprint} and \cite{Karma1, Karma2}. Crucial for the convergence analysis are the above listed criteria (C1)-(C4), that $\MSL$ satisfies a $T$-G\r{a}rding's inequality with respect to the splitting in (C1)~\cite{SISC}, and that the mapping  $\C\setminus\{0\}\ni\kappa\mapsto\MSL$ is holomorphic~\cite[Lem.~2]{UngerPreprint}.  In~\cite{Halla:2016,UngerPreprint}  sufficient conditions for the convergence of conforming Galerkin methods are specified  for eigenvalue problems for holomorphic $T$-G\r{a}rding  operator-valued functions.  According to~\cite[Lem.~5]{UngerPreprint} and \cite[Lem.~2.7]{Halla:2016} the  criteria (C1)-(C4) allow  the application of the  convergence theory established in~\cite{Karma1, Karma2} to the Galerkin eigenvalue problem~\eqref{GalerkinEVP}. From the comprehensive convergence results presented in~\cite{Karma1, Karma2,Halla:2016,UngerPreprint}  we only want to state the error estimate for semi-simple eigenvalues $\kappa$:
\begin{equation}\label{ErrorEstimate1}
 \vert \kappa-\kappa_h\vert =\mathcal{O}(\delta_{m,p}^2),
\end{equation}
where
\begin{equation}\label{ErrorEstimate2}
 \delta_{m,p}
:=\sup_{\substack{\pmb j\in \mathrm{ker}\MSL\\
\|\pmb j\|_{\pmb H^{-1/2}_\times(\div_\Gamma,\Gamma) }= 1}} \inf_{\pmb j_h\in \pmb{\mathbb S}^p_m(\Gamma)}\|\pmb j-\pmb j_h\|_{\pmb H^{-1/2}_\times(\div_\Gamma,\Gamma) }.
\end{equation}
The quadratic convergence order with respect to $\delta_{h,p}$ follows from the fact that for any eigenfunction $\pmb  j_\text{adj}$  of  the adjoint eigenvalue problem there exists an eigenfunction $\pmb j$ of the eigenvalue problem  $\MSL\pmb j=0$ such that  $\overline{\pmb j}=\pmb j_\text{adj}$, see~\cite[Lem.~3]{UngerPreprint}.
The estimate \eqref{eq::iga::approximation} together with \eqref{ErrorEstimate1} and \eqref{ErrorEstimate2} yields the final estimate
\begin{align}
 \vert \kappa-\kappa_h\vert = \mathcal{O}(h^{2(p+1/2)}),\label{ErrorEstimate3}
\end{align}
for sufficiently smooth surface current densities $\pmb j.$

\begin{remark}[Convergence on non-smooth geometries]\label{rem::nonsmoothGeom}
    In this sense sufficiently smooth needs to be understood in terms of patchwise regularity as explained in \cite{NuMa}. In general, for surface current densities on non-smooth geometries this smoothness assumption will not be fulfilled, since $\pmb j$ may admit singularities at corners and edges.

    However, for the specific densities of resonant modes within interior cavities this smoothness assumption will often be fulfilled.  For the cube an analytical representation of $\pmb j$ is known to be smooth~\cite{WienersXin:2013}, and even for other, non-trivial geometries this can be argued.
\end{remark}
\section{The Contour Integral Method}\label{sec::cim}

This chapter will give a short summary of the contour integral method as introduced by Beyn \cite{Beyn_2012aa}, without any non-trivial modifications. Note that there are alternative approaches by other authors, where the first publication \cite{Asakura_2009aa} seems to go back to 2009, compare also \cite{Imakura_2016aa}. After this short review, we will discuss three numerical examples that utilise this method to solve Problem \ref{GalerkinEVP}.

	Let $\bb T\colon M\to \mathbb C^{m\times m}$ be holomorphic on some domain $M\subset \mathbb C$. Let $\bb T^h$ denote the Hermitian transpose of $\bb T$ and $\bb T'$ denote its usual transpose. 
	We want to solve the nonlinear eigenvalue problem
	$$ \bb T(z)\bb v  = 0,\quad \bb v\in\mathbb C^m,\quad \bb v\neq \bb 0,\quad z\in M.$$ 
	If $z$ is no eigenvalue of $\bb T$, we find by the condition above that the kernel of $\bb T(z)$ is trivial, thus $\bb T(z)$ has full rank and is invertible.

	Fix some domain $D\subseteq  M$ with boundary $\partial D,$ and assume there to be $k$ eigenvalues $(\lambda_j)_{j\leq k}$ in the interior of $D$, such that all eigenvalues are simple. We remark that {the entire approach  was also generalised to that case of non simple eigenvalues \cite{Beyn_2012aa}.} 
	As a consequence of the version of the Keldysh Theorem stated in \cite{Mennicken_2003aa} one can proof the following result, on which this method is based.

	\begin{theorem}[Contour Integral Theorem, {\citep[Thm.~2.9.]{Beyn_2012aa}}]\label{13}\index{contour integral}
		For $\bb T$ as above and holomorphic $f\colon M\to \mathbb C$ it holds that 
		\begin{align*}
			\frac 1{2\pi i}\int_{\partial D} f(z) \bb T(z)^{-1} \mathrm{d} z =\sum^k_{j=1}f(\lambda_j)\bb v_h\bb w_h^h,
		\end{align*}
		where $\bb v_j$ and $\bb w_j$ are left and right eigenvectors corresponding to $\lambda_n$ that are normalized according to 
		$$\bb w_j^h \bb T'(\lambda_j)\bb v_j=1$$
		for all $j\leq k$.
	\end{theorem}
	
	For our purposes, where we set $\bb T(z)=\bb V_z$ as a matrix-valued function, cf.~\eqref{eq::linear_system}, the contour integral method can be reduced to  the following theorem.

	\begin{theorem}[Linearisation of Eigenvalue Problems, {\cite[Thm.~3.1]{Beyn_2012aa}}] \index{contour integral!method}\label{cim::thm::cim}	Suppose that $\bb T\colon  M\to\mathbb C^{m\times m}$ is holomorphic and has only simple eigenvalues $(\lambda_j)_{j\leq k}$ in some connected subdomain $D\subset M$.  
	Then there exists a diagonalizable matrix $\bb B$ which can be computed from evaluations of $\bb T$ such that $\bb B$ that has the same eigenvalues as the eigenvalue problem under consideration within $D$.
	\end{theorem}

	The proof of this theorem is constructive and corresponds to the following algorithm, on which our implementation is based.

\begin{algo}[Contour Integral Method, {\cite[p.~15]{Beyn_2012aa}}]\label{algo}
Let $\bb T$ be given as above, $\lbrace t_n\rbrace_{n\leq N}$ be a discretisation of some boundary $\partial D$ as above, $\delta,\epsilon>0$ and  $\ell<\text{size}(\bb T(z))$.\\
\begin{algorithm}[H]
 \KwData{$\lbrace t_n\rbrace_{n\leq N},\bb T,\delta,\epsilon,\ell$}
 \KwResult{$(\lambda_j)_{j\leq k}$}
 kfound $\gets$ false\;
 $m\gets$ size($\bb T$)\;
 \While{kfound $==$ false }{
 $\hat {\bb V} \gets \operatorname{RandomFullRank}(m\times \ell)$\;
 $\bb A_0 \gets \frac 1{iN}\sum_{j=1}^{N}\bb (T( t_j))^{-1}\hat {\bb V}$\;
 $({\bb V},{\bb \Sigma},{\bb W}^h)\gets$ SVD$({\bb A}_0)$ \tcp*{$\bb \Sigma =\text{diag}(\sigma_1,\dots,\sigma_\ell)$}
 $k\gets j,$ where $\sigma_1\geq\dots\geq\sigma_j>\delta>\sigma_{j+1}\approx\dots\approx\sigma_\ell \approx 0 $\;
 \eIf{$k < \ell$}{
 	kfound $\gets$ true\;
 	$\bb V_0 \gets {\bb V}(1\colon m,1\colon k)$\;
 	$\bb W_0 \gets {\bb W}(1\colon \ell,1\colon k)$\;
 	${\bb \Sigma}_0\gets \bb\Sigma$\;}
 	{$\ell\gets \ell+1$\;}
 }
 $\bb A_1 \gets \frac 1{iN}\sum_{j=1}^{N} t_j\bb T( t_j)^{-1}\hat{\bb V}$\;
 $\bb B\gets  \bb V_0^h\bb A_1\bb W_0\bb \Sigma_0^{-1}$\;

 $(\lambda_j)_{j\leq k}\gets {eigs}(B)$\;
\end{algorithm}
\end{algo}

At first the algortihm might seem prohibitively expensive due to the application of a singular value decomposition. However, since the number of eigenvalues $k$ will be small and often a reasonable upper estimate $\ell_0$ of the number of eigenvalues is known such that one can  choose $\ell=\ell_0$.   
Then the complexity of the singular value decomposition becomes negligible in comparison to solving the linear system in lines 5 and 15 of Algorithm~\ref{algo}.

Moreover, in an actual implementation $\bb A_1$ and $\bb A_0$ are assembled simultaneously, since the most expensive operation of the algorithm is evaluating $\bb T(t_j)^{-1}\hat {\bb V}$. In general, $k$ and $\ell$ will be small, so storing both $\bb A_j$ poses no issues. Since smooth contours should be used, exclusively, one can expect the trapezoidal rules for the assembly of the $\bb A_j$ to converge exponentially with respect to $N$. 
Thus, the bottleneck in terms of accuracy of the entire scheme is the accuracy in which $\bb X$ represents the bilinear form of the EFIE.

\begin{remark}[Obtaining the Eigenvectors]
Note that through this algorithm we not only obtain the eigenvalues of Problem \ref{problem::variational}, but also the coefficients of the corresponding eigenfunctions in the form of the matrices $\bb V$ and $\bb W.$ Note moreover that the number of non-zero singular values reflects the number of solutions within $D$ and can be used for verification of an implementation if the number of solutions is known from analytical representations or measurements.
\end{remark}


\section{Numerical Examples}\label{sec::num}

In the following we will discuss some numerical experiments showcasing the application of the contour integral method to the isogeometric boundary element method. The particular implementation used is \emph{Bembel}, which is available open-source \cite{bembel,SoftwareX}. A branch containing the code to recreate all of the presented numerical examples is also available \cite{CodeBranch}.

Since no quasi-optimal preconditioners for the isogeometric discretisation of the  electric field integral equation are known, iterative solvers yield suboptimal runtimes. Thus, the following numerical experiments will use a dense matrix assembly together with the partially pivoted LU decomposition of the linear algebra library Eigen~\cite{eigen3} to solve the arising systems. The utilised higher-order approaches yield systems small enough for this approach to be more than feasible.

\begin{figure}\centering
    \begin{tikzpicture}[scale = .7]
        \begin{axis}[
        height = 7.5cm,
        width=8cm,
        ymode=log,
        xlabel=$h$,
        xtick={0, 1, 2, 3, 4, 5},
        xticklabels={1, {1}/{2}, {1}/{4}, {1}/{8}, {1}/{16}, {1}/{32}},
        ylabel=error,
        legend pos=south west,
        legend columns=2,
        ]
\addplot+[line width=1.5pt,mark size =2.5pt,red] table [trim cells=true,x=M,y=sum_err] {data/cim_gpu1/Acim_cube_1_25};
\addlegendentry{ $p=1$}
\addplot+[line width=1.5pt,mark size =2.5pt,blue] table [trim cells=true,x=M,y=sum_err] {data/cim_gpu1/Acim_cube_2_25};
\addlegendentry{ $p=2$}
\addplot+[line width=1.5pt,mark size =2.5pt,brown] table [trim cells=true,x=M,y=sum_err] {data/cim_gpu1/Acim_cube_3_25};
\addlegendentry{ $p=3$}
\addplot+[line width=1.5pt,mark size =2.5pt,gray] table [trim cells=true,x=M,y=sum_err] {data/cim_gpu1/Acim_cube_4_25};
\addlegendentry{ $p=4$}

\convfitorder{sum_err}{data/cim_gpu1/cim_cube_1_30}{3}{{$\mathcal O(h^3)$}}
\convfitorder{sum_err}{data/cim_gpu1/cim_cube_2_30}{5}{{$\mathcal O(h^5)$}}
\convfitorder{sum_err}{data/cim_gpu1/cim_cube_3_30}{7}{{$\mathcal O(h^7)$}}
\convfitorder{sum_err}{data/cim_gpu1/cim_cube_4_30}{9}{{$\mathcal O(h^9)$}}
        \end{axis}
        \end{tikzpicture}\hfill
        \begin{tikzpicture}[scale = .7]
        \begin{axis}[
        height = 7.5cm,
        width=8cm,
        ymode=log,
        xlabel=$i$,
        xtick = {0,1,2,3,4,5,6,7,8,9},
        xticklabels={1,2,3,4,5,6,7,8,9,10},
        ylabel=$\sigma_i$,
        ]
\addplot+[only marks,line width=1.5pt,mark size =2.5pt,red] table [trim cells=true,x index = {0},y index = {1}] {data/cim_gpu1/Asvd_ev_cim_cube_3_25_0};
\addlegendentry{$h=1$}
\addplot+[only marks,line width=1.5pt,mark size =2.5pt,blue] table [trim cells=true,x index = {0},y index = {1}] {data/cim_gpu1/Asvd_ev_cim_cube_3_25_1};
\addlegendentry{$h=1/2$}
\addplot+[only marks,line width=1.5pt,mark size =2.5pt,brown] table [trim cells=true,x index = {0},y index = {1}] {data/cim_gpu1/Asvd_ev_cim_cube_3_25_2};
\addlegendentry{$h=1/4$}
\addplot+[only marks,line width=1.5pt,mark size =2.5pt,gray] table [trim cells=true,x index = {0},y index = {1}] {data/cim_gpu1/Asvd_ev_cim_cube_3_25_3};
\addlegendentry{$h=1/8$}
        \end{axis}
        \end{tikzpicture}
\caption{On the left the minimal difference of the simultaneously computed first and second eigenvalue $\lambda_\mathrm{cim}$ of the cube to their analytical solution $\lambda_i$ via $\min_{i=0,1}\vert\lambda_\mathrm{cim}-\lambda_i\vert$. On the right the computed singular values for the example with $p=3.$}\label{fig::cim::cube}
\end{figure}
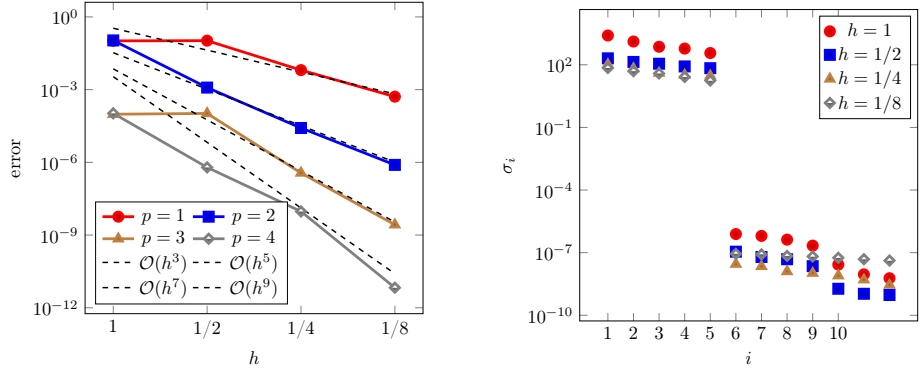

As a first example, we investigate the first two eigenvalues of the unit cube, where an analytical solution is given by the eigenvalues $\lambda_{\mathrm{ana},0}=\pi\sqrt{2}$ of multiplicity three and  $\lambda_{\mathrm{ana},1}=\pi\sqrt{3}$ of multiplicity two. The ellipse was defined as $$\sin(t) + i\cdot0.05\cdot\cos(t)+5.0,\quad\text{ for }t=[0,2\pi),$$ where the contour integrals were evaluated, again, with $N=25.$ The error 
visualised corresponds to 
$$\mathrm{error} = \frac 15\sum_{0\leq j<5}\min_{i=0,1}\vert\lambda_j-\lambda_{\mathrm{ana},i}\vert.$$
As one can see in Figure~\ref{fig::cim::cube}, the multiplicity of the eigenvalues is reflected perfectly by the non-zero singular values, i.e., all eigenvalues have been recognised. Moreover, we have the theoretically obtainable convergence order of ${\mathcal O}(h^{2p+1})$.

\begin{figure}\centering
    \begin{tikzpicture}[scale = .7]
        \begin{axis}[
        height = 7.5cm,
        width=8cm,
        ymode=log,
        xlabel=$h$,
        xtick={0, 1, 2, 3, 4, 5},
        xticklabels={1, {1}/{2}, {1}/{4}, {1}/{8}, {1}/{16}, {1}/{32}},
        ylabel=error,
        legend pos=south west,
        legend columns=2,
        ]
\addplot+[line width=1.5pt,mark size =2.5pt,red] table [trim cells=true,x=M,y=sum_err] {data/cim_gpu2/Acim_sphere_int_1_25};
\addlegendentry{ $p=1$}
\addplot+[line width=1.5pt,mark size =2.5pt,blue] table [trim cells=true,x=M,y=sum_err] {data/cim_gpu2/Acim_sphere_int_2_25};
\addlegendentry{ $p=2$}
\addplot+[line width=1.5pt,mark size =2.5pt,brown] table [trim cells=true,x=M,y=sum_err] {data/cim_gpu2/Acim_sphere_int_3_25};
\addlegendentry{ $p=3$}
\addplot+[line width=1.5pt,mark size =2.5pt,gray] table [trim cells=true,x=M,y=sum_err] {data/cim_gpu2/Acim_sphere_int_4_25};
\addlegendentry{ $p=4$}
\convfitorder{sum_err}{data/cim_gpu2/cim_sphere_int_1_40}{3}{{$\mathcal O(h^3)$}}
\convfitorder{sum_err}{data/cim_gpu2/cim_sphere_int_2_40}{5}{{$\mathcal O(h^5)$}}
\convfitorder{sum_err}{data/cim_gpu2/cim_sphere_int_3_40}{7}{{$\mathcal O(h^7)$}}
\convfitorder{sum_err}{data/cim_gpu2/cim_sphere_int_4_40}{9}{{$\mathcal O(h^9)$}}
        \end{axis}
        \end{tikzpicture}\hfill
        \begin{tikzpicture}[scale = .7]
        \begin{axis}[
        height = 7.5cm,
        width=8cm,
        ymode=log,
        xlabel=$i$,
        xtick = {0,1,2,3,4,5,6,7,8,9},
        xticklabels={1,2,3,4,5,6,7,8,9,10},
        ylabel=$\sigma_i$,
        ]
\addplot+[only marks,line width=1.5pt,mark size =2.5pt,red] table [trim cells=true,x index = {0},y index = {1}] {data/cim_gpu2/Asvd_ev_cim_sphere_int_3_25_0};
\addlegendentry{$h=1$}
\addplot+[only marks,line width=1.5pt,mark size =2.5pt,blue] table [trim cells=true,x index = {0},y index = {1}] {data/cim_gpu2/Asvd_ev_cim_sphere_int_3_25_1};
\addlegendentry{$h=1/2$}
\addplot+[only marks,line width=1.5pt,mark size =2.5pt,brown] table [trim cells=true,x index = {0},y index = {1}] {data/cim_gpu2/Asvd_ev_cim_sphere_int_3_25_2};
\addlegendentry{$h=1/4$}
\addplot+[only marks,line width=1.5pt,mark size =2.5pt,gray] table [trim cells=true,x index = {0},y index = {1}] {data/cim_gpu2/Asvd_ev_cim_sphere_int_3_25_3};
\addlegendentry{$h=1/8$}
        \end{axis}
        \end{tikzpicture}
\caption{On the left the mean absolute difference of the computed first eigenvalue of the sphere to the analytical solution. On the right the computed singular values for the example with $p=3.$}\label{fig::cim::sphere}
\end{figure}
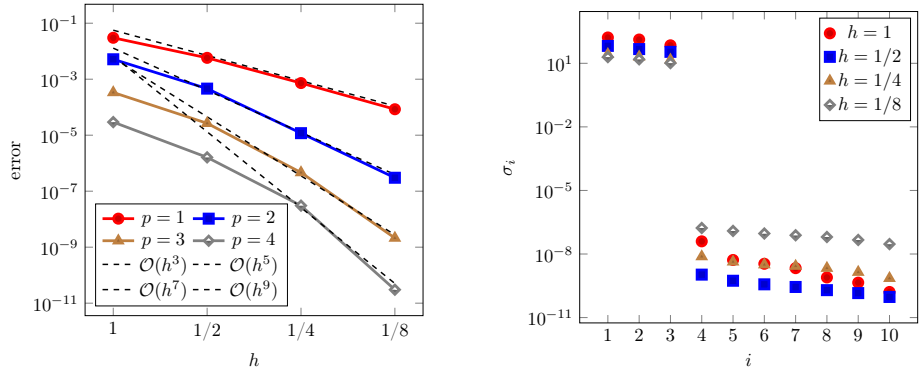

\begin{figure}
\begin{tikzpicture}[scale=.75]
\begin{axis}[
width=.65\textwidth,
xlabel = real part,
ylabel = imaginary part,
]
\addplot[red,only marks,mark = o,very thick,mark options={solid},mark size=3pt] coordinates {
    (2.743707269992269, 0)
};
\addplot[blue,only marks,mark = x,very thick,mark options={solid}] coordinates {
    (2.7437123, 0)
};
\addplot[dashed,mark = o,mark options={solid},gray] coordinates {

( 2.984292e+00, 1.243449e-02)
( 2.938153e+00, 2.408768e-02)
( 2.864484e+00, 3.422736e-02)
( 2.767913e+00, 4.221640e-02)
( 2.654508e+00, 4.755283e-02)
( 2.531395e+00, 4.990134e-02)
( 2.406309e+00, 4.911436e-02)
( 2.287110e+00, 4.524135e-02)
( 2.181288e+00, 3.852566e-02)
( 2.095492e+00, 2.938926e-02)
( 2.035112e+00, 1.840623e-02)
( 2.003943e+00, 6.266662e-03)
( 2.003943e+00, -6.266662e-03)
( 2.035112e+00, -1.840623e-02)
( 2.095492e+00, -2.938926e-02)
( 2.181288e+00, -3.852566e-02)
( 2.287110e+00, -4.524135e-02)
( 2.406309e+00, -4.911436e-02)
( 2.531395e+00, -4.990134e-02)
( 2.654508e+00, -4.755283e-02)
( 2.767913e+00, -4.221640e-02)
( 2.864484e+00, -3.422736e-02)
( 2.938153e+00, -2.408768e-02)
( 2.984292e+00, -1.243449e-02)
( 3, -1.224647e-17)

};
\end{axis}
\end{tikzpicture}
\hfill
\begin{tikzpicture}[scale=.75]
\begin{axis}[
width=.65\textwidth,
xlabel = real part,
ylabel = imaginary part,
]
\addplot[red,only marks,mark = o,very thick,mark options={solid}, mark size=3pt] coordinates {
    (4.44288293815836, 0)
    (5.44139809270265, 0)
};
\addlegendentry{Analytical eigenvalue};
\addplot[blue,only marks,mark = x,very thick,mark options={solid}] coordinates {
    (4.442882, 0)
    (5.441338, 0)
};
\addlegendentry{Computed eigenvalue};
\addplot[gray,dashed,mark = o,mark options={solid}] coordinates {

( 5.968583e+00, 1.243449e-02)
( 5.876307e+00, 2.408768e-02)
( 5.728969e+00, 3.422736e-02)
( 5.535827e+00, 4.221640e-02)
( 5.309017e+00, 4.755283e-02)
( 5.062791e+00, 4.990134e-02)
( 4.812619e+00, 4.911436e-02)
( 4.574221e+00, 4.524135e-02)
( 4.362576e+00, 3.852566e-02)
( 4.190983e+00, 2.938926e-02)
( 4.070224e+00, 1.840623e-02)
( 4.007885e+00, 6.266662e-03)
( 4.007885e+00, -6.266662e-03)
( 4.070224e+00, -1.840623e-02)
( 4.190983e+00, -2.938926e-02)
( 4.362576e+00, -3.852566e-02)
( 4.574221e+00, -4.524135e-02)
( 4.812619e+00, -4.911436e-02)
( 5.062791e+00, -4.990134e-02)
( 5.309017e+00, -4.755283e-02)
( 5.535827e+00, -4.221640e-02)
( 5.728969e+00, -3.422736e-02)
( 5.876307e+00, -2.408768e-02)
( 5.968583e+00, -1.243449e-02)
( 6, -1.224647e-17)

};
\addlegendentry{Quad. pts. $t_j$}
\end{axis}
\end{tikzpicture}
\caption{The setup for the contour integral method. Sphere (left) and cube (right), both computed with $N=25$, $p=2$, $h=1/4$.
}\label{fig::cim::setup}

\end{figure}
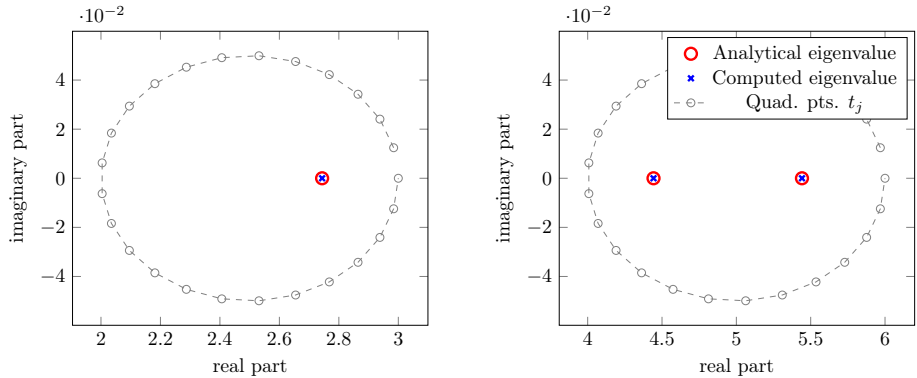

\subsection{Comparison to IGA-FEM and Commercial Tools}

As a second example, we compute the first eigenvalue for the sphere, cf.~Figure~\ref{fig::cim::sphere} for the results. 
The contour was chosen as the curve $$ 0.5\cdot\sin(t) + i\cdot0.05\cdot\cos(t)+2.5,\quad\text{ for }t=[0,2\pi).$$ The trapezoidal rule for lines 5 and 15 of the algorithm was chosen with $N=25$. 
A closed-form solution $\lambda_{\mathrm{ana}}$ is known \cite[Sec.~5.1.2]{UngerPreprint} and given as the first root of the spherical Bessel function of the first kind, cf.~Figure~\ref{fig::cim::setup}.
It is an eigenvalue of multiplicity three, thus three non-zero singular values in $\bb \Sigma$ are expected. This behaviour is reflected by the numerical examples perfectly. The error in Figure \ref{fig::cim::sphere} refers to the average error of all three computed eigenvalues, i.e., 
$$ \mathrm{error} = \frac 13\sum_{0\leq j < 3} \vert \lambda_j - \lambda_{\mathrm{ana}}\vert.$$
The convergence behaviour w.r.t.~$h$ matches the orders ${\mathcal O}(h^{2p+1})$ predicted by~\eqref{ErrorEstimate3} once more.

We have also computed approximations of the smallest eigenvalues of the unit sphere with the volume-based IGA software package  \emph{GeoPDEs 3.0}~\cite{geopde3,Falco_2011aa}. In Table~\ref{Tab:FEMBEM} we showcase results for different polynomial degrees and refinements the number of degrees of freedom and the reached accuracy  for the IGA-BEM and the volume-based IGA. One can observe that  for the IGA-BEM a 
notably fewer (at least $20\times$) number of degrees of freedom are necessary for each polynomial degree in order to reach the same accuracy as  for the volume-based IGA. 
The difference to commercial tools, in this case CST Microwave Studio 2018, is even more pronounced, cf.~Table~\ref{Tab:CST}.
However, the matrices for the IGA-BEM are dense and the eigenvalue problem non-linear.

\begin{table}
\begin{footnotesize}\begin{center}
\begin{tabular}{r r r r| r r r r }
\toprule
\multicolumn{4}{c}{IGA-BEM}&\multicolumn{4}{c}{volume-based IGA
} \\
$p$ & $m$ & DOFs & error & $p$ & SD & DOFs & error\\
\midrule
3 & 1  &192  & 2.66e-05         &3  &4  &4350   &3.345e-05\\
3 & 2  &432  &4.63e-07          &3  &8  &20450  &3.45e-07\\
3 & 3  &1200 &2.12e-09          &3  &14 & 84560  &1.10e-08\\
\midrule
4  &1  &300  &1.62e-06          &4  &4  &6944   &2.36e-06\\
4  &2  &588  &3.05e-08          &4  &8  &27280  &3.94e-09\\
4  &3  &1452 &3.89e-11          &4  &13 &84560  &7.88e-11\\
\midrule        
5  &1  &432  &8.97e-08          &5  &4   &10408 &1.88e-07\\
5  &2  &768  &2.04e-09          &5  &8   &35484 &7.33e-11\\
5  &3  &1728 &7.17e-12          &5  &11  &69600 &2.1498e-12\\
\bottomrule
\end{tabular}
\end{center}\end{footnotesize}
\caption{Error and degrees of freedom (DOFs) for the approximation of the smallest eigenvalue of the sphere for different polynomial degrees $p$, refinement levels $m$  and number of subdivisions~(SD).}
\label{Tab:FEMBEM} 
\end{table}

\begin{table}
\begin{footnotesize}\begin{center}
\begin{tabular}{r r r r}
\toprule
\multicolumn{4}{c}{CST Microwave Studio 2018}\\
elements & SD & DOFs & error \\
\midrule
1089 & 6&21597 & 1.8638e-07\\
5292 & 10&101997 & 5.4204e-09\\
9191 & 12&175818 & 1.9524e-09\\
\bottomrule
\end{tabular}
\end{center}\end{footnotesize}
\caption{Error and degrees of freedom (DOFs) for the approximation of the smallest eigenvalue of the sphere with CST Microwave Studio 2018. Settings are 
Mesh: Tetrahedral, Desired accuracy: 1e-12, Solver order: 3rd (constant), Curved elements: up to order 5 (user defined)}
\label{Tab:CST} 
\end{table}

\subsection{An Industrial Application: TESLA Cavity}

As a third example we discuss an eigenvalue computation of the one-cell TESLA geometry as shown in Figures~\ref{fig::intro::tesla1} and \ref{fig::intro::tesla2}, with the results depicted in Figure~\ref{fig::cim::tesla1}. For this example no analytical solution is known, but experience dictates a resonant frequency around $\kappa \approx 26.5$. We choose the contour $$0.5\cdot\sin(t) + i\cdot0.01\cdot\cos(t)+26.5,\quad\text{ for }t=[0,2\pi)$$ and $N=8.$ As a reference solution we utilise the result of a computation with $p=5$ and $h=1/8.$ 
This reference solution was compared to a computation with CST Microwave Studio 2018. The set solver order was 3rd (constant) and the mesh was generated with 
200\,771 curved tetrahedral elements of order 5. CST yields the solution of 1.27666401260 GHz. Our reference solution of $\lambda_{\mathrm{cim}} = 26.75690023$ corresponds to a frequency of 1.276664064 GHz. This results in a relative error of 4.018e-08. Thus our experiments are in good agreement with those of the CST software. {However, note that in Figure \ref{fig::cim::tesla1} one can clearly see stagnation w.r.t.~the CST Solution on higher-order computations and small $h$. This suggests that the Bembel reference solution is more accurate, provided the convergence of the contour integral approach behaves as observed in the previous numerical experiments.}

The order of convergence for the cavity is not as pronounced as in the examples with analytical solution, which, in light of the estimate~\eqref{ErrorEstimate3} and Remark~\ref{rem::nonsmoothGeom} is most likely due to the reduced regularity of the corresponding eigenfunction.
Either way, one can clearly see an increased accuracy in higher-order approaches.

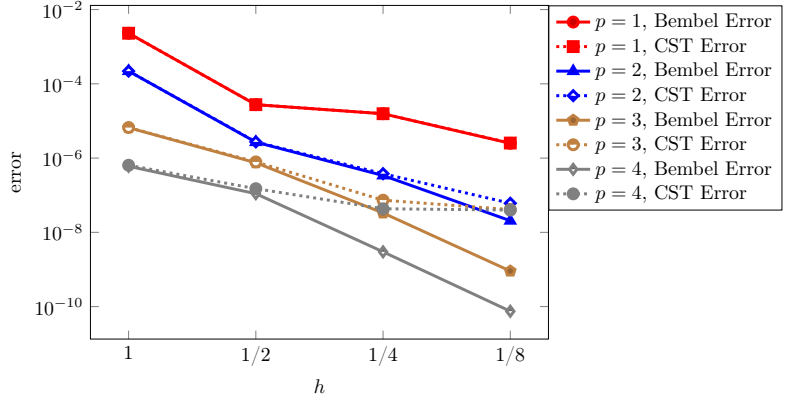
\begin{figure}[t!]\centering
    \begin{tikzpicture}[scale = .75]
        \begin{axis}[
        height = 7.5cm,
        ymode=log,
        xlabel=$h$,
        xtick={0, 1, 2, 3, 4, 5},
        xticklabels={1, {1}/{2}, {1}/{4}, {1}/{8}, {1}/{16}, {1}/{32}},
        ylabel=error,
        legend style={at={(1,1)},anchor=north west,legend columns=1,cells={anchor=west},}
        ]
        \addplot+[line width=1.5pt,mark size =2.5pt,red] table [trim cells=true,x=M,y=est_err] {data/1cell/cim_tesla_1_cell_1_22_gerhard};
        \addlegendentry{ $p=1$, Bembel Error}
        \addplot+[line width=1.5pt,mark size =2.5pt,mark options = {solid}, dotted,red] table [trim cells=true,x=M,y=cst_err] {data/1cell/cim_tesla_1_cell_1_22_gerhard};
        \addlegendentry{ $p=1$, CST Error}
        \addplot+[line width=1.5pt,mark size =2.5pt,blue] table [trim cells=true,x=M,y=est_err] {data/1cell/cim_tesla_1_cell_2_22_gerhard};
        \addlegendentry{ $p=2$, Bembel Error}
        \addplot+[line width=1.5pt,mark size =2.5pt,mark options = {solid}, dotted,blue] table [trim cells=true,x=M,y=cst_err] {data/1cell/cim_tesla_1_cell_2_22_gerhard};
        \addlegendentry{ $p=2$, CST Error}
        \addplot+[line width=1.5pt,mark size =2.5pt,brown] table [trim cells=true,x=M,y=est_err] {data/1cell/cim_tesla_1_cell_3_22_gerhard};
        \addlegendentry{ $p=3$, Bembel Error}
        \addplot+[line width=1.5pt,mark size =2.5pt,mark options = {solid}, dotted,brown] table [trim cells=true,x=M,y=cst_err] {data/1cell/cim_tesla_1_cell_3_22_gerhard};
        \addlegendentry{ $p=3$, CST Error}
        \addplot+[line width=1.5pt,mark size =2.5pt,gray] table [trim cells=true,x=M,y=est_err] {data/1cell/cim_tesla_1_cell_4_22_gerhard};
        \addlegendentry{ $p=4$, Bembel Error}
        \addplot+[line width=1.5pt,mark size =2.5pt,mark options = {solid}, dotted,gray] table [trim cells=true,x=M,y=cst_err] {data/1cell/cim_tesla_1_cell_4_22_gerhard};
        \addlegendentry{ $p=4$, CST Error}
        \end{axis}
        \end{tikzpicture}
        \caption{Eigenvalue problem of the one-cell TESLA cavity solved for the first eigenvalue. Error has to be understood as the relative difference to the result of a computation. Bembel Reference refers to the error of our method to a reference computation with $h=1/8$ and $p=5$ of our own implementation, and the CST Reference to the error w.r.t.~a reference solution obtained via CST Microwave Studio 2018.}\label{fig::cim::tesla1}
\end{figure}
\section{Conclusion}\label{sec::con}

Overall, the contour integral method yields exceptional accuracies in conjunction with our isogeometric boundary element method.
Judging from the asymptotic behaviour predicted in \eqref{ErrorEstimate3} and observed in the numerical examples, the accuracy of the combination of isogeometric boundary element method and the contour integral method promises higher orders of convergence to the correct solution compared to currently implemented volume based approaches, since these will not benefit from the convergence order of ${\mathcal O}(h^{2p+1})$. 
This means that, for the same computational resources, the maximum reachable accuracy of the IGA-BEM with contour integration is higher compared to many volume-based methods. However, due to the long time spend solving the systems, this means that the IGA-BEM with contour integral offers this higher accuracy only in exchange for runtime, until efficient preconditioning strategies are available, thus making the utilisation of fast methods and iterative solvers viable for this type of application. 

\begin{quote}
\footnotesize
{\textbf{Acknowledgement:} The work of Gerhard Unger is supported by 
the Austrian Science Fund (FWF), project P31264.
The work of Felix Wolf is supported by DFG Grants SCHO1562/3-1 and KU1553/4-1, the Excellence
Initiative of the German Federal and State Governments and the Graduate School of Computational Engineering at TU Darmstadt.}
\end{quote}


\begin{thebibliography}{10}
\bibitem{CodeBranch}
F.~Wolf, Fork of {B}embel, will be made available upon acceptance of the
  manuscript.
\end{thebibliography}

\end{document}